\documentclass[runningheads]{llncs}
\usepackage[T1]{fontenc}
\usepackage{newtxtext}       % 
\usepackage[varvw]{newtxmath}       % selects Times Roman as basic font

\usepackage{amsmath}
\usepackage{algorithm}
\usepackage{algpseudocode}
\usepackage{graphicx}
\usepackage{url}
\usepackage{xparse}
\usepackage{tikz}
\usetikzlibrary{arrows.meta,positioning}

\NewDocumentEnvironment{informaltheorem}{m}
{\par\medskip\noindent\textbf{Theorem~\ref{#1} (Informal). }\itshape}
{\par\medskip}

\NewDocumentEnvironment{informalcor}{m}
{\par\medskip\noindent\textbf{Corollary~\ref{#1} (Informal). }\itshape}
{\par\medskip}

\newenvironment{proofsketch}
{\par\noindent\textit{Proof sketch.}\ }
{\hfill\qed\par}

\newcommand{\Proj}{\mathsf{Project}}
\newcommand{\OPT}{\mathsf{OPT}}
\newcommand{\Cmax}{C_{\max}}
\newcommand{\move}{\mathsf{move}}

\newcommand{\Z}{\mathbb{Z}}
\newcommand{\Ind}{\mathbf{1}}
\newcommand{\Err}{\mathcal{E}}
\newcommand{\ErrJ}{\mathcal{E}_J}
\newcommand{\DP}{\mathsf{DP}}
\newcommand{\cost}{\mathsf{cost}}
\newcommand{\Out}{\mathsf{Out}}
\newcommand{\In}{\mathsf{In}}
\newcommand{\RepairOracle}{\mathsf{GlobalRepairOracle}}
\newcommand{\ExactRepairOracle}{\mathsf{ExactGlobalRepairOracle}}
\newcommand{\bigO}{\mathcal{O}}
\newcommand{\moveJobs}{\mathsf{moveJobs}}
\newcommand{\RepairOracleJ}{\mathsf{GlobalRepairOracle}_J}

\begin{document}
\title{From Estimates to Schedules: 
Learning-Augmented Restricted Assignment}
\author{Michalis Xefteris}
\authorrunning{M. Xefteris}
\institute{University of Bremen, Germany\\
\email{xefteris@uni-bremen.de}}
\maketitle              % typeset the header of the contribution

\begin{abstract}
In this work, we study \textsc{Restricted Assignment} scheduling on multiple machines, where each job can be processed only on a specified subset of machines and the objective is to minimize the makespan. We consider a learning-augmented setting in which the predicted assignment is not required to be feasible. The prediction error (moved-load) is measured by the total processing volume that must be reassigned in order to obtain an optimal feasible schedule.

Using a single prediction, we obtain two types of guarantees. First, we design an algorithm whose approximation ratio degrades smoothly with the prediction error while retaining a worst-case guarantee independent of the prediction quality. More precisely, for any fixed constant, we can make the additive dependence on the prediction error arbitrarily small, at the cost of increasing the polynomial running time. This guarantee can also be combined with any approximation algorithm for the problem without predictions to obtain robustness.

Second, given a makespan estimate, we provide a repair procedure that returns a schedule matching this estimate in time parameterized by the prediction error. This allows the algorithm to exploit the separation between estimation and approximation algorithms for \textsc{Restricted Assignment}. Finally, we complement the repair algorithm with a parameterized hardness result, showing that exact moved-load repair with a given target makespan is $\mathrm{W[1]}$-hard when parameterized by the amount of moved-load.
\end{abstract}

\section{Introduction}
Recent advances in machine learning (ML) have stimulated a growing interest in using predictions within algorithm design. Learning from historical data can provide useful insight into different aspects of an optimization problem, such as structural properties of an instance, features of a high-quality solution, partial solutions, or near-optimal algorithmic decisions. This has motivated the study of \emph{learning-augmented algorithms} (or \emph{algorithms with predictions}) which aim to leverage such information to achieve improved performance when the predictions are accurate, while still maintaining rigorous worst-case guarantees when they are not. In this setting, predictions about the problem at hand are typically provided in a black-box manner.

This line of research originated in the work of Mahdian et al.~\cite{est} on online advertising allocation and was later formalized by
Lykouris and Vassilvitskii~\cite{lykouris} for online caching. The goals that
have since become standard are that an algorithm should be (i)~\emph{consistent},
performing near-optimally when predictions are accurate, (ii)~\emph{robust}, preserving
near-optimal worst-case guarantees when predictions are poor, and (iii)~\emph{smooth},
with performance degrading gradually as prediction error increases. The field has blossomed with applications to online algorithms, streaming and sketching algorithms, data structures and combinatorial optimization. An almost complete list of papers in this area can be found in~\cite{alps}.

Algorithms with predictions have also been used to improve approximation algorithms for NP-hard problems, with clustering emerging as one of the main directions. For $k$-means clustering, learning-augmented algorithms
that use predicted cluster labels have been developed in~\cite{ergun,gamlath,nguyenimproved},
achieving a $(1+\bigO(\alpha))$-approximation when the per-cluster label
error is at most~$\alpha$. Hierarchical clustering has also been studied in the learning-augmented setting~\cite{hierarchical}.
Antoniadis et al.~\cite{antoniadis} studied a broad class of NP-hard selection problems, using possibly infeasible predictions to obtain improved approximation guarantees when the predictions are sufficiently accurate. Another well-studied direction concerns Max-CSP problems, including Max-Cut. Bampis et al.~\cite{bampis} considered dense Max-CSPs and showed how predictions can be used to speed up polynomial-time approximation schemes (PTAS).
Max-Cut has also been investigated in a learning-augmented setting with noisy predictions that are slightly better than random guesses.
In two independent works, Cohen-Addad et al.~\cite{cohen-addad} and Ghoshal et al.~\cite{makarychev} improved the approximation guarantees for Max-Cut in this setting. The same model has also been applied to several other NP-hard problems~\cite{braverman,bampis_2,aamand}, yielding better guarantees than those achievable without predictions.

In this work, we study offline scheduling on unrelated parallel machines in the learning-augmented setting.  Scheduling is a fundamental optimization problem with applications in manufacturing, transportation, and computing systems. More specifically, we consider the \textsc{Restricted Assignment} problem, in which each job can be processed only on a specified subset of machines, and the objective is to minimize the makespan. For the problem without predictions, no polynomial-time algorithm can achieve a worst-case approximation ratio better than $3/2$ unless $\mathrm{P}=\mathrm{NP}$~\cite{LST90}.

While learning-augmented algorithms for online scheduling have been studied extensively~\cite{online_1,online_2}, the offline setting has received much less attention. Bampis et al.~\cite{bampis_2} studied a single-machine scheduling problem with the objective of minimizing the sum of completion times. Given predictions about the relative order of jobs that are correct with probability slightly greater than $1/2$, they designed an exact polynomial-time algorithm for the problem. Moseley et al.~\cite{moseley} studied the problem of scheduling stochastic jobs on a single machine when the scheduler is given only predicted job-size distributions, which may be only approximately correct, rather than the true distributions. 

\subsection{Contributions}

In this work, we are given a predicted assignment of $n$ jobs to $m$ machines. As we show in Section~\ref{sec:approx}, feasibility can be restored by a simple projection without increasing the prediction error.
The prediction error $\Err$ is just the minimum total processing time (moved-load) that must be reassigned from the predicted job-to-machine assignment in order to obtain a feasible schedule with optimal makespan $\OPT$. If there are several optimal solutions, $\Err$ is the minimum such total among them. A formal description of the prediction model is given in Section~\ref{sec:prelim}.

Our first goal is to improve the approximation ratio of the problem using the possibly infeasible prediction $\hat\pi$. To this end, we first project $\hat\pi$ into a feasible schedule $\hat\sigma$, obtaining a basic projection guarantee with makespan at most $\OPT+\Err$.
Building on this projection, we derive our smoothness guarantee by guessing the few large jobs on which the prediction differs from an optimal schedule and solving the assignment LP for the remaining jobs.
More precisely, for every fixed $\delta\in(0,1]$, there is an algorithm running in time
$(nm)^{\bigO(1/\delta)}$ that outputs a feasible schedule $\sigma$ with makespan at most $\OPT+\delta\Err$.
We then robustify this approach by
running in parallel an $\alpha$-approximation algorithm for the problem without
predictions and returning the better of the two schedules, obtaining the
following guarantee.
\begin{informaltheorem}{theo:rob}
For every fixed $\delta\in(0,1]$, given a predicted assignment $\hat\pi$ and an $\alpha$-approximation algorithm for \textsc{Restricted Assignment} without predictions, one can output a schedule $\sigma$ with makespan
$\Cmax(\sigma)\le \min\left\{\OPT+\delta \Err,\,\alpha \cdot \OPT\right\}$ in time $(nm)^{\bigO(1/\delta)}$.
\end{informaltheorem}

A main novelty of our work is that we use the same prediction and the same error measure, in two distinct ways. First, we use it to obtain smooth and robust approximation guarantees. Second, we use the same prediction to obtain repair algorithms whose running time is parameterized by the prediction error. Prior learning-augmented approximation algorithms for offline NP-hard problems typically use predictions either to improve the approximation guarantee in polynomial time, or to speed up a known algorithmic scheme. In contrast, our framework connects these two uses of predictions within a single model. More specifically, we show that, given an estimate $T$ of the optimal makespan, one can construct a repair procedure that takes the prediction $\hat\pi$ and outputs a schedule with makespan at most $T$ in time parameterized by the prediction error $\Err$ (Section~\ref{sec:repair-load}). At a high level, the algorithm first guesses the prediction error via a doubling search, then identifies by exhaustive enumeration the subset of machines that must be repaired, and finally solves the repair problem on a fixed set of machines using a dynamic programming approach inspired by multi-dimensional knapsack. 
\begin{informalcor}{cor:time-load}
There is a repair procedure that, given a makespan threshold $T\ge\OPT$ and a predicted assignment $\hat\pi$, returns a feasible schedule $\sigma'$ for \textsc{Restricted Assignment} with
$\Cmax(\sigma')\le T$ in time
$\widetilde{\bigO}\!\left(
n \cdot m^{\bigO(\Err)} \cdot \Err^{\bigO(\Err)}
\right)$, where $\widetilde{\bigO}(\cdot)$ hides logarithmic factors.
\end{informalcor}

We complement this algorithmic result with a parameterized hardness result. Exact moved-load repair for \textsc{Restricted Assignment}, with a given target makespan, is $\mathrm{W[1]}$-hard when parameterized by the amount of moved-load (Section~\ref{subsec:repair-hardness}). Thus, under standard parameterized complexity assumptions, the error-dependent running time cannot in general be replaced by an $\mathrm{FPT}$ dependence on the moved-load alone.
Together, these results show that both parts of our framework are close to best possible: the smoothness guarantee can make the coefficient of $\Err$ arbitrarily small at the cost of a $(nm)^{\bigO(1/\delta)}$ running time, while exact repair is unlikely to admit an $\mathrm{FPT}$ dependence on the prediction error alone.

\paragraph*{Estimation vs Approximation.}
This repair approach is useful when the makespan estimate $T$ yields a better makespan than the robust guarantee of Theorem~\ref{theo:rob}. 
For many NP-hard problems studied in the literature, however, the best estimation ratio achievable in polynomial time coincides with the best approximation ratio, and therefore such an estimate does not lead to any additional benefit. \textsc{Restricted Assignment} is one of the few known exceptions.\footnote{Under standard complexity assumptions, it has been shown that there are (artificial) NP-hard optimization problems for which estimation is easier than approximation~\cite{feige}.}
Recall that a polynomial-time algorithm has estimation ratio $\rho$ if it computes a value $T$ within a factor $\rho$ of the optimum ($\OPT \le T \le \rho \cdot \OPT$). For \textsc{Restricted Assignment}, the best known approximation ratio is $\alpha=2$~\cite{LST90}, whereas the best known estimation ratio is $11/6$~\cite{lars_2}. In Section~\ref{sec:applications}, we instantiate our guarantees for \textsc{Restricted Assignment} and for its two-valued variant, in which each job has one of two possible processing times.

At the same time, the moved-load error can be pessimistic when the prediction makes only a few mistakes on large jobs. To address this, we develop a second repair framework whose running time depends on the number of moved jobs (moved-job error $\ErrJ$) rather than on the moved-load error (see Appendix~\ref{sec:repair_jobs}).
The aforementioned hardness construction also implies that exact moved-job repair for \textsc{Restricted Assignment} is $\mathrm{W[1]}$-hard when parameterized by the number of moved jobs.
For \textsc{Restricted Assignment}, the method follows the same repair framework as before. Compared with the moved-load repair, it incurs an additional $\varepsilon T$ loss in the makespan, arising from the necessary rounding of the processing times. For the two-valued case, we refine this approach by separating big and small jobs: we use a dynamic program for the big jobs and a min-cost flow formulation for the small jobs, which removes the $(1+\varepsilon)$ loss and yields an exact repair with makespan at most $T$. Both repair procedures have similar worst-case running-time guarantees as in Corollary~\ref{cor:time-load}.

\paragraph*{About the prediction model.}
In line with many works on learning-augmented offline optimization, our algorithms rely on having access to predictions for a computationally hard optimization problem. It is not yet clear whether current ML methods can reliably provide accurate predictions about such problems, though a growing body of work in machine learning for combinatorial optimization seeks to learn structural information about high-quality solutions~\cite{bengio,combinatorial}. In our setting, predictions need not be feasible schedules, since feasibility can be enforced by a simple projection without increasing the prediction error. This allows the prediction to capture assignment information without necessarily satisfying all eligibility constraints. Therefore, in application domains with recurring structure, such predictions may be realistic in practice.

\paragraph*{Concurrent work.}
Concurrently and independently of our work, Baba et al.~\cite{baba2026} studied learning-augmented approximation algorithms for makespan minimization on unrelated machines. Their approach is based on predictions of assignments of big jobs to machines, where big jobs are defined through a threshold parameter $\varepsilon$. They obtain a $(1+\varepsilon)\mathrm{OPT}$ makespan guarantee when the predicted assignments of big jobs are accurate, and for imperfect predictions their guarantee incurs an additional additive loss depending on false-negative and false-positive errors in these predicted assignments, while retaining the classical $2\mathrm{OPT}$ robustness guarantee. 
They also use an extra parameter controlling a brute-force search to recover a bounded number of false-negative big-job assignments. Their framework differs from ours in both the prediction/error model and the algorithmic use of predictions. In particular, our work focuses on \textsc{Restricted Assignment} with moved-load/moved-job error measures, and uses the prediction not only for smooth and robust approximation, but also for error-dependent repair from a makespan estimate, exploiting the separation between estimation and approximation for \textsc{Restricted Assignment}.

\section{Preliminaries} \label{sec:prelim}

We consider \textsc{Restricted Assignment} scheduling on $m$ parallel machines and $n$ jobs. Each job $j\in[n]$ has processing time $p_j>0$ and a set of eligible machines $\Gamma(j)\subseteq[m]$.
A schedule is feasible if $\sigma(j)\in\Gamma(j)$ for all jobs.
For a feasible schedule $\sigma$, we define
$L_i(\sigma) = \sum_{j: \sigma(j)=i} p_j$ and
$\Cmax(\sigma) = \max_{i\in[m]}L_i(\sigma)$.
The goal of the problem is to find a feasible schedule that minimizes the makespan. 

Let $\sigma^*$ be a feasible schedule that achieves the optimal makespan $\OPT$. More specifically,
$\OPT = \min\{\Cmax(\sigma): \sigma \text{ feasible}\}$.
Throughout the paper, we assume that every input instance of the problem is feasible, i.e., $\Gamma(j)\neq\emptyset$ for all jobs $j$. Moreover, we assume that all processing times are integers: $p_j\in\Z_{\ge 1}$ for all jobs $j$. If processing times are rational, one can scale them to integers, which affects the moved-load
error $\Err$ (and hence the corresponding running times) only by the
scaling factor and preserves the moved-job error $\ErrJ$. Since all processing times are integers, we may assume without loss of generality that the makespan threshold $T$ is an integer.

\subsection{Prediction Model}
For every instance of the problem, we are given a predicted assignment $\hat\pi:[n]\to[m]$ of jobs to machines, which may be infeasible, i.e., $\hat\pi(j)\notin\Gamma(j)$ for some jobs.

Next, we define the prediction error used in this work\footnote{We also consider an alternative error measure based on the number of moved jobs. Its formal definition is deferred to Appendix~\ref{sec:repair_jobs}.}. Intuitively, it measures the total load (accumulated processing time) that must be reassigned to transform a predicted job-to-machine mapping into an (optimal) assignment. First, let us define the moved-load.

\begin{definition}[Moved-load] \label{def:load}
For a feasible schedule $\sigma$ and any mapping $\tau:[n]\to[m]$ (feasible or infeasible), define:
\[
M(\sigma,\tau) = \{j\in[n]: \sigma(j)\neq \tau(j)\},
\qquad
\move(\sigma,\tau) = \sum_{j\in M(\sigma,\tau)} p_j.
\]
\end{definition}

Now we are ready to define the prediction error with respect to an optimal schedule for an instance of the problem. 

\begin{definition}[Prediction error wrt $\OPT$] \label{def:error}
The prediction error of a predicted assignment $\hat\pi$ with respect to an optimal solution is
\[
 \Err = \Err(\OPT) = \min\{\move(\sigma^*,\hat\pi): \sigma^* \text{ feasible and } \Cmax(\sigma^*)=\OPT\}.
\]
\end{definition}

We will also need to define the prediction error required to obtain a feasible schedule of makespan at most $T$. Throughout the paper, we assume that $T \ge \OPT$.

\begin{definition}[Prediction error wrt $T$]
Fix $T$. The minimum moved-load needed to reach makespan at most $T$ from a predicted assignment $\hat\pi$ is
\[
\Err(T) =\min\{\move(\sigma,\hat\pi): \sigma \text{ feasible and } \Cmax(\sigma)\le T\}.
\]
\end{definition}

\section{Error-Dependent Approximation Guarantees}
\label{sec:approx}

In this section, we use the prediction to obtain approximation guarantees that degrade smoothly with the moved-load error $\Err$. We first project the possibly infeasible prediction to a feasible schedule $\hat\sigma=\Proj(\hat\pi)$, which already gives a basic projection guarantee. We then combine this projection with a guessing-and-rounding procedure to obtain our smoothness guarantee.

The first step is very simple: we project $\hat\pi$ to a feasible schedule.
Specifically, for a fixed arbitrary ordering $\pi_g$ of the jobs, the procedure $\Proj_{\pi_g}(\hat\pi)$ processes jobs one by one. If $\hat\pi(j)\in\Gamma(j)$, it keeps the predicted assignment, and otherwise it reassigns $j$ to the least-loaded eligible machine (breaking ties arbitrarily). The output is a feasible schedule $\hat\sigma$.

In the rest of the paper, we omit explicit reference to the order $\pi_g$ and write
$\hat\sigma=\Proj(\hat\pi)$, using a fixed arbitrary order. Moreover, unless stated otherwise, we use the projected predicted schedule $\hat\sigma$.

\subsection{Smooth Approximation}

We begin with the projection step. First, we prove that $\Proj(\hat\pi)$ does not increase the moved-load error.

\begin{lemma}[Projection monotonicity (moved-load)]\label{lem:proj}
For every feasible schedule $\sigma$, we have $\move(\sigma,\hat\sigma)\le \move(\sigma,\hat\pi)$, where $\hat\sigma=\Proj(\hat\pi)$.
\end{lemma} 

\begin{proof}
Fix a feasible schedule $\sigma$ and a job $j$.

If $\hat\pi(j)\in\Gamma(j)$, then $\hat\sigma(j)=\hat\pi(j)$, and thus
$
\Ind[\sigma(j)\neq \hat\sigma(j)] = \Ind[\sigma(j)\neq \hat\pi(j)],
$
so job $j$ contributes the same amount to $\move(\sigma,\hat\sigma)$ and $\move(\sigma,\hat\pi)$.

If $\hat\pi(j)\notin\Gamma(j)$, then feasibility implies $\sigma(j)\in\Gamma(j)$ and hence $\sigma(j)\neq \hat\pi(j)$, i.e.,
$\Ind[\sigma(j)\neq \hat\pi(j)]=1$, while $\Ind[\sigma(j)\neq\hat\sigma(j)]\in\{0,1\}$.
Thus $\Ind[\sigma(j)\neq\hat\sigma(j)]\le \Ind[\sigma(j)\neq \hat\pi(j)]$ for every job $j$.
Multiplying by $p_j$ and summing over jobs yields $\move(\sigma,\hat\sigma)\le \move(\sigma,\hat\pi)$. 
\qed \end{proof}

We are now ready to show that the projected predicted schedule $\hat\sigma=\Proj(\hat\pi)$  has makespan close to $\OPT$ when the prediction error with respect to an optimal solution is small.

\begin{theorem}[Smoothness of $\Proj$]\label{thm:smooth-move}
Let $\hat\sigma=\Proj(\hat\pi)$. Then,
$\Cmax(\hat\sigma)\le \OPT + \Err$.
\end{theorem}

\begin{proof}
Let $\sigma^*$ be an optimal schedule.
Consider any machine $i$.

Let $\hat L_i = L_i(\hat\sigma)$ and
$X_i = \{j : \hat{\sigma}(j) = i,\, \sigma^*(j) \neq i\}$ be the jobs assigned to machine~$i$ under~$\hat{\sigma}$ but not under~$\sigma^*$. Since every job on machine $i$ under $\hat{\sigma}$ either also appears on $i$ under $\sigma^*$ or belongs to $X_i$, we have $\hat{L}_i \leq L_i(\sigma^*) + \sum_{j \in X_i} p_j$, and thus $\hat{L}_i - L_i(\sigma^*) \leq \sum_{j \in X_i} p_j$.

Then,
$\hat L_i - L_i(\sigma^*) \le \sum_{j\in X_i} p_j \le \move(\sigma^*,\hat\sigma)$.
Since $L_i(\sigma^*)\le \OPT$, we obtain $\hat L_i\le \OPT + \move(\sigma^*,\hat\sigma)$ for every machine $i$, and thus
$\Cmax(\hat\sigma)\le \OPT + \move(\sigma^*,\hat\sigma)$.

By Lemma~\ref{lem:proj}, $\move(\sigma^*,\hat\sigma)\le \move(\sigma^*,\hat\pi)$.
Minimizing $\move(\sigma^*,\hat\pi)$ over optimal schedules yields
$\Cmax(\hat\sigma)\le \OPT + \Err$. 
\qed \end{proof}

We now prove our smoothness guarantee. The idea is to guess a
constant number of large jobs on which the projected prediction differs from an optimal schedule closest to the prediction in moved-load, and then use the standard assignment LP and rounding for \textsc{Restricted Assignment} on the remaining, small jobs.

\begin{theorem}[Smoothness]\label{thm:smooth-improved}
For every fixed $\delta\in(0,1]$, there is an algorithm running
in time $(nm)^{\bigO(1/\delta)}$ that outputs a feasible schedule $\sigma$
with $\Cmax(\sigma)\le \OPT+\delta\Err$.
\end{theorem}

\begin{proof}
Fix $\delta\in(0,1]$ and set
$q=\left\lceil\frac{1}{\delta}\right\rceil-1$.

Compute $\hat\sigma=\Proj(\hat\pi)$. The algorithm enumerates a threshold
$\theta\in\{0\}\cup\{p_j:j\in[n]\}$, a set
$B\subseteq\{j:p_j>\theta\}$ with $|B|\le q$, and an assignment
$\varphi:B\to[m]$ with $\varphi(j)\in\Gamma(j)$ for all $j\in B$.
For a fixed guess, jobs in $B$ are fixed according to $\varphi$, jobs
$j\notin B$ with $p_j>\theta$ are fixed to $\hat\sigma(j)$, and all jobs of
size at most $\theta$ are left free.

For the free jobs, we solve the standard assignment LP with the fixed loads
already assigned. If the LP is feasible, let $T'$ be its optimum value and apply
the Lenstra--Shmoys--Tardos (LST) rounding~\cite{LST90}. By the additive form of
the LST rounding, the resulting schedule has makespan at most $T'+\theta$,
because every free job has size at most $\theta$. If the LP is infeasible, we
discard this guess. The algorithm returns the best feasible schedule over all guesses.

Let $\sigma^*$ be an optimal schedule minimizing $\move(\sigma^*,\hat\pi)$, so
$\move(\sigma^*,\hat\pi)=\Err$. By Lemma~\ref{lem:proj},
$\move(\sigma^*,\hat\sigma)\le \Err$.

Let $R=M(\sigma^*,\hat\sigma)$ be the set of jobs on which $\sigma^*$ differs
from $\hat\sigma$.

If $|R|\le q$, then the algorithm considers the guess $\theta=0$, $B=R$, and
$\varphi(j)=\sigma^*(j)$ for all $j\in B$. Since there are no free jobs when $\theta=0$, this guess reconstructs
$\sigma^*$. Hence the algorithm returns a schedule of makespan at most
$\OPT\le \OPT+\delta\Err$.

Assume now that $|R|>q$. Order the jobs of $R$ as $j_1,\ldots,j_{|R|}$ so that
$p_{j_1}\ge p_{j_2}\ge\dots \ge p_{j_{|R|}}$, and set
$\theta=p_{j_{q+1}}$. Then
$(q+1)\theta
    \le
    \sum_{\ell=1}^{q+1}p_{j_\ell}
    \le
    \move(\sigma^*,\hat\sigma)
    \le
    \Err$,
and therefore $\theta\le \Err/(q+1)$.

Consider the guess with $B=\{j\in R:p_j>\theta\}$ and
$\varphi(j)=\sigma^*(j)$ for all $j\in B$. By the choice of $\theta$, we have
$|B|\le q$, so this guess is enumerated. Moreover, every job of size larger
than $\theta$ outside $B$ is not in $R$, and hence is assigned in the same way
by $\sigma^*$ and $\hat\sigma$. Thus all fixed jobs are consistent with
$\sigma^*$, while all remaining jobs have size at most $\theta$ and are left
free. Therefore, $\sigma^*$ is feasible for the LP of this guess, so
$T'\le \OPT$. The rounded schedule for this guess has makespan at most
$T'+\theta
    \le
    \OPT+\frac{\Err}{q+1}
    \le
    \OPT+\delta\Err$,
where the last inequality follows from $1/(q+1)\le\delta$.

The number of guesses is at most $(n+1)n^qm^q$, and each guess requires one LP
solve and one rounding step. Thus the running time is
$(nm)^{\bigO(q)}=(nm)^{\bigO(1/\delta)}$.
\qed \end{proof}

\subsection{Robust Approximation}

If the smoothness guarantee is not sufficient, we can combine our algorithm with any $\alpha$-approximation algorithm for the problem that does not use predictions.

\begin{theorem}[Smoothness \& Robustness] \label{theo:rob}
Let $\mathcal{A}_{\alpha}$ be any $\alpha$-approximation algorithm for \textsc{Restricted Assignment} that does not use predictions. 
For every fixed $\delta\in(0,1]$, given a predicted assignment $\hat\pi$, one can output a schedule $\sigma$ in time $(nm)^{\bigO(1/\delta)}$ such that
$\Cmax(\sigma)\le \min\{\OPT+\delta\Err,\alpha \cdot \OPT\}$.
\end{theorem}

\begin{proof}
Run the algorithm of Theorem~\ref{thm:smooth-improved} and
$\mathcal{A}_\alpha$ in parallel and return the schedule with the smaller
makespan. The makespan guarantee follows immediately.
\qed \end{proof}

\section{Error-Dependent Repair} \label{sec:repair-load}

In this section, we show how to repair the projected feasible schedule $\hat\sigma$ in time that depends on the prediction error. Assume we are given an estimate $T$ for the optimal makespan of \textsc{Restricted Assignment} satisfying
$\OPT \le T < \alpha \cdot \OPT$, where $\alpha$ is the best known approximation ratio for the problem without predictions.
When the makespan estimate $T$ improves over the robust guarantee, we present an $\mathrm{XP}$ (slice-wise polynomial) repair procedure that starts from the projected schedule $\hat\sigma$ and
outputs a schedule of makespan at most $T$ in error-dependent time. Finally, we complement this result with a parameterized hardness result for exact repair with a given target makespan, showing that an $\mathrm{FPT}$ algorithm parameterized only by the prediction error would imply $\mathrm{W[1]}=\mathrm{FPT}$ (Section~\ref{subsec:repair-hardness}).

\subsection{Lower bounds on prediction error}
We next derive some lower bounds on the prediction error, which will be useful in the design of our learning-augmented repair algorithm. Recall that $\Err$ and $\Err(T)$ are defined with respect to the original prediction $\hat{\pi}$. In this section, we repair the projected feasible schedule $\hat{\sigma}=\Proj(\hat{\pi})$.
For every machine $i$, let $\hat L_i=L_i(\hat\sigma)$ denote its load under the projected predicted schedule.
Whenever we compare a feasible schedule $\sigma$ to $\hat{\sigma}$, we use Lemma~\ref{lem:proj}, which gives $\move(\sigma,\hat{\sigma}) \le \move(\sigma,\hat{\pi})$.

First, we show that achieving makespan at most $T$ requires no more correction than reaching an optimal schedule. More formally, we have the following.

\begin{lemma} \label{lem:lb-err(T)}
If $T\ge \OPT$, then $\Err(T)\le \Err$.
\end{lemma}

\begin{proof}
Let $\sigma^*$ be an optimal schedule, so $\Cmax(\sigma^*)=\OPT\le T$. Therefore, $\sigma^*$ is  feasible for the minimization defining $\Err(T)$. Hence, $\Err(T) \le \move(\sigma^*,\hat\pi)$.
Minimizing the right-hand side over all optimal schedules gives $\Err(T)\le \Err$. 
\qed \end{proof}

We define the total overload under $\hat\sigma$ with respect to makespan $T$ as
$\Delta(T) = \sum_{i\in[m]} (\hat L_i - T)^+$,
where $(\hat L_i - T)^+ = \max\{\hat L_i - T, 0\}$ for every $i \in [m]$.
Let us now prove the following proposition.

\begin{proposition}\label{prop:lb-chain}
For every $T\ge\OPT$, $\Err(T)\ge \Delta(T)$.
\end{proposition}

\begin{proof}
Let $\sigma$ be feasible with $\Cmax(\sigma)\le T$. For each machine $i$, the jobs
that $\hat\sigma$ places on $i$ but $\sigma$ does not, must have total processing time
at least $(\hat L_i - T)^+$. Summing over machines gives $\move(\sigma, \hat\sigma) \geq \Delta(T)$.
By Lemma~\ref{lem:proj}, $\move(\sigma, \hat\pi) \geq \move(\sigma, \hat\sigma) \geq \Delta(T)$. Minimizing over $\sigma$ yields $\Err(T)\ge\Delta(T)$.
\qed \end{proof}

\subsection{Repair Oracle} \label{subsec:repair}
In this subsection, we describe a repair procedure at a given makespan threshold $T$ and moved-load budget $K$. The procedure relies on two ingredients:
(i) a dynamic program that repairs $\hat\sigma$ when the set of incident machines is fixed, and (ii) a global enumeration over all candidate incident sets of bounded size.
Let us first present some lemmas that will prove useful. 

First, we define the incident set
$A_{\mathrm{inc}}(\sigma,\hat\sigma)
= \{\hat\sigma(j): j\in M(\sigma,\hat\sigma)\} \cup \{\sigma(j): j\in M(\sigma,\hat\sigma)\}$.

\begin{lemma}[Overloaded machines must be incident]\label{lem:over-incident}
If $\sigma$ is feasible and $\Cmax(\sigma)\le T$, then every overloaded machine wrt $T$ is incident:
$\{i:\hat L_i>T\} \subseteq A_{\mathrm{inc}}(\sigma,\hat\sigma)$.
\end{lemma}

\begin{proof}
Fix $i$ with $\hat L_i>T$. If $i\notin A_{\mathrm{inc}}(\sigma,\hat\sigma)$ then no predicted job leaves $i$ and no job is
moved to $i$, hence $L_i(\sigma)=\hat L_i>T$, contradicting $\Cmax(\sigma)\le T$. 
\qed \end{proof}

Next, we bound the number of machines touched by moved jobs in terms of the moved-load budget.

\begin{lemma}[Incident set size bound]\label{lem:inc-size}
Let $K\in\Z_{\ge 0}$.
If $\move(\sigma,\hat\sigma)\le K$, then
$|A_{\mathrm{inc}}(\sigma,\hat\sigma)|\le 2K$.
\end{lemma}

\begin{proof}
Since processing times are integers and $p_j\ge 1$ for all jobs, each moved job contributes at least $1$ to moved-load, so
$|M(\sigma,\hat\sigma)|\le \move(\sigma,\hat\sigma)\le K$.
Each moved job contributes at most two machines (origin and destination) to $A_{\mathrm{inc}}$, so
$|A_{\mathrm{inc}}(\sigma,\hat\sigma)|\le 2|M(\sigma,\hat\sigma)|\le 2K$. 
\qed \end{proof}

Let the set of machines overloaded under $\hat\sigma$ wrt $T$ be
$O(T) = \{i\in[m]: \hat L_i > T\}$.

We are now ready to describe a Repair oracle for a fixed incident set of machines using dynamic programming. The DP follows the standard pattern of multi-dimensional knapsack: we process items one by one, and for each item we choose one of a few possible destinations. Each choice updates a small vector that summarizes the partial reassignment (e.g., signed load changes across the incident set), and we pay a cost only when the choice differs from the predicted assignment. The DP therefore searches, within a bounded state space, for a minimum-cost partial reassignment that satisfies all feasibility constraints.

\begin{lemma}[Fixed-set Repair oracle (moved-load)]
\label{lem:fixed-dp-load}
Fix a threshold makespan $T$, a moved-load budget $K\in\mathbb{Z}_{\ge 1}$, and a feasible predicted schedule $\hat\sigma$.
Let $A\subseteq[m]$ satisfy $O(T)\subseteq A$. Write $A=\{i_1,\dots,i_a\}$ with $a=|A|\ge 1$ and define $J_A = \{j\in[n]: \hat\sigma(j)\in A\}$.
There exists a dynamic program (DP) running in time
$O\!\left(|J_A|\cdot a \cdot (2K+1)^{a-1}\right)$
that decides whether there exists a feasible schedule $\sigma$ such that:
\begin{enumerate}
\item $\sigma(j)=\hat\sigma(j)$ for all $j\notin J_A$;
\item $\sigma(j)\in \Gamma(j)\cap A$ for all $j\in J_A$;
\item $\move(\sigma,\hat\sigma)\le K$;
\item $\Cmax(\sigma)\le T$.
\end{enumerate}
If such a schedule exists, the DP outputs one. Otherwise it outputs \textsc{Fail}.
\end{lemma}

\begin{proofsketch}
    Since $O(T)\subseteq A$, machines outside $A$ already have load at most $T$ and remain unchanged. We process the jobs of $J_A$ one by one. Each DP state keeps track of how much total load each machine in $A$, except one, has gained or lost relative to $\hat\sigma$ after the jobs processed so far. The load change on the remaining machine is determined by load conservation.

    For each state, the DP stores the minimum moved-load needed to obtain it. Assigning job $j$ to an eligible machine $u\in\Gamma(j)\cap A$ removes $p_j$ from machine $\hat\sigma(j)$ and adds $p_j$ to machine $u$, and incurs cost $p_j$ iff $u\neq\hat\sigma(j)$. Since the total moved-load is at most $K$, every net load change lies in $[-K,K]$, so there are at most $(2K+1)^{a-1}$ states. A final state is accepted if the resulting load of every machine in $A$ is at most $T$. The DP runs in $O\!\left(|J_A|\cdot a \cdot (2K+1)^{a-1}\right)$ time. The full proof is given in Appendix~\ref{app:moved-load-details}.
\end{proofsketch}

We now wrap Lemma~\ref{lem:fixed-dp-load} into a global oracle by enumerating all candidate sets $A$ that could contain the incident machines of a solution of moved-load at most $K$.
If $|O(T)|>2K$, then no such set exists and the oracle immediately returns \textsc{Fail}.

\begin{definition}[Global Repair oracle with $(T,K)$ (moved-load)]
Given makespan $T$, budget $K\in\mathbb{Z}_{\ge 1}$ and a feasible predicted schedule $\hat\sigma$, 
$\RepairOracle(T,K,\hat\sigma)$ enumerates all sets $A\subseteq[m]$ such that $O(T)\subseteq A$ and
$|A|\le 2K$, runs the DP of Lemma~\ref{lem:fixed-dp-load} for each such $A$, and returns the first schedule found (or \textsc{Fail} if none exists).
\end{definition}

\begin{lemma}[Correctness of the Global Repair oracle (moved-load)]
\label{lem:oracle-correct}
If there exists a feasible schedule $\sigma'$ such that $\Cmax(\sigma')\le T$ and
$\move(\sigma',\hat\sigma)\le K$, then $\RepairOracle(T,K,\hat\sigma)$ returns a feasible schedule $\sigma$ with $\Cmax(\sigma)\le T$ and $\move(\sigma,\hat\sigma)\le K$.
\end{lemma}

The proof is given in Appendix~\ref{app:moved-load-details}.
Next, we will bound the running time of $\RepairOracle(T,K,\hat\sigma)$.
The overhead comes from enumerating all candidate incident sets $A$ with $O(T)\subseteq A$ and $|A|\le 2K$, and running the fixed-set DP (Lemma~\ref{lem:fixed-dp-load}) for each such set. Again, the proof is deferred to Appendix~\ref{app:moved-load-details}.

\begin{lemma}[Running time of the Global Repair oracle] 
\label{lem:oracle-runtime-load}
Fix $T$, $K\in\mathbb{Z}_{\ge 1}$ and a feasible predicted schedule $\hat\sigma$.
Then, $\RepairOracle(T,K,\hat\sigma)$ has running time
$O\!\left(n\cdot m^{2K}\cdot (2K+1)^{2K+1}\right)$.
\end{lemma}

\subsection{Global Repair}
In this subsection, we describe the Global Repair procedure (see Algorithm~\ref{alg:global-repair-load}) that does not require knowing the optimal
moved-load budget ($\Err(T)$) in advance. The algorithm guesses the necessary moved-load by doubling a budget parameter $K$,
starting from the lower bound $\Delta(T)$, and invokes a bounded-budget oracle to obtain a schedule of makespan at most~$T$.
The resulting running time depends smoothly on the prediction error $\Err$.

\begin{algorithm}[htbp]
\caption{\textsc{GlobalRepair}$(T,\hat\pi)$}
\label{alg:global-repair-load}
\begin{algorithmic}[1]
\Require Integer makespan threshold $T\ge \OPT$, possibly infeasible prediction $\hat\pi$
\Ensure A feasible schedule $\sigma$ with $\Cmax(\sigma)\le T$
\State $\hat\sigma \gets \Proj(\hat\pi)$
\State Compute loads $\hat L_i$ for all $i\in[m]$, and $\Delta(T)\gets \sum_{i\in[m]} (\hat L_i - T)^+$
\If{$\Delta(T)=0$}
 \State \Return $\hat\sigma$
\EndIf
\State $K \gets \Delta(T)$
\While{\textbf{true}}
  \State $\sigma \gets \RepairOracle(T,K,\hat\sigma)$
  \If{$\sigma \neq \textsc{Fail}$}
    \State \Return $\sigma$
  \EndIf
  \State $K \gets 2K$
\EndWhile
\end{algorithmic}
\end{algorithm}

\begin{theorem}[Global Repair at makespan $T$]
\label{thm:parametric-repair}
Fix any makespan threshold $T\ge \OPT$ and let $\hat\sigma=\Proj(\hat\pi)$.
Algorithm~\ref{alg:global-repair-load} returns a feasible schedule $\sigma$ with $\Cmax(\sigma)\le T$.
Moreover, if $\Delta(T)>0$ it stops after reaching some integer $K$ satisfying
$K < 2\Err(T) \le 2\Err$, and the number of oracle calls is at most
$\left\lceil \log_2\!\left(\frac{\Err(T)}{\Delta(T)}\right)\right\rceil + 1\ \le\ \left\lceil \log_2\!\left(\frac{\Err}{\Delta(T)}\right)\right\rceil + 1$.

If $\Delta(T)=0$, it directly outputs the predicted schedule $\hat\sigma$.
\end{theorem}

\begin{proofsketch}
If $\Delta(T)=0$, then $\hat\sigma$ already has makespan at most $T$. Otherwise, Proposition~\ref{prop:lb-chain} gives $\Delta(T)\le\Err(T)$, so doubling reaches a first value $K^*\ge\Err(T)$ with $K^*<2\Err(T)\le2\Err$. By the definition of $\Err(T)$ and Lemma~\ref{lem:proj}, there is a schedule of makespan at most $T$ whose moved-load from $\hat\sigma$ is at most $K^*$. Hence Lemma~\ref{lem:oracle-correct} guarantees that the repair oracle succeeds by this iteration. Therefore, the stopping budget satisfies $K\le K^*<2\Err(T)\le2\Err$. The bound on the number of calls follows from the doubling search. The full proof is given in Appendix~\ref{app:moved-load-details}.
\end{proofsketch}

Consequently, since the successful oracle call uses $K < 2\Err$, we get the following corollary. Note that a naive approach that explicitly guesses the set of moved jobs would lead to a running time with an $n^{\Theta(\Err)}$-type dependence.

\begin{corollary}[Running time of \textsc{GlobalRepair}] \label{cor:time-load}
Assume that $T\ge \OPT$ and $\Delta(T)>0$. Then Algorithm~\ref{alg:global-repair-load} returns a feasible schedule $\sigma$ in time
$\widetilde{\bigO}\!\left(
n \cdot m^{\bigO(\Err)} \cdot \Err^{\bigO(\Err)}
\right)$ with
$\Cmax(\sigma)\le T$.
\end{corollary}

\subsection{Hardness for Exact Moved-Load Repair}
\label{subsec:repair-hardness}

We now complement the moved-load repair algorithm with a parameterized hardness result. In particular, exact repair with a given target makespan does not admit an $\mathrm{FPT}$ algorithm parameterized by the amount of moved-load, unless $\mathrm{W[1]}=\mathrm{FPT}$. The proof is by a reduction from \textsc{Multicolored Clique} and is deferred to Appendix~\ref{app:repair-hardness}.

\begin{definition}[Exact Moved-Load Repair]
Given an instance of \textsc{Restricted Assignment}, a feasible predicted schedule $\hat\sigma$, a makespan threshold $T$, and an integer budget $K$, decide whether there exists a feasible schedule $\sigma$ such that $\Cmax(\sigma)\le T$ and $\move(\sigma,\hat\sigma)\le K$.
\end{definition}

\begin{theorem}
\label{thm:exact-moved-load-repair-hardness}
\textsc{Exact Moved-Load Repair} is $\mathrm{W[1]}$-hard parameterized by $K$, even when $\hat\sigma$ is feasible and the instance has at most three distinct processing times.
\end{theorem}

\section{Applications} \label{sec:applications}
In this section, we instantiate the guarantees of Sections~\ref{sec:approx} and~\ref{sec:repair-load} using the best known approximation and estimation bounds for \textsc{Restricted Assignment} and its two-valued special case. In the following, we assume that all processing times are integers.

\subsection{Restricted Assignment}

For \textsc{Restricted Assignment}, the classical LST algorithm has an approximation ratio of $2$~\cite{LST90}. More precisely, this was later improved to $2-\frac{1}{m}$~\cite{better-2}, where $m$ denotes the number of machines, but the worst-case ratio remains $2$ asymptotically as $m$ grows. On the estimation side, Jansen and Rohwedder gave a polynomial-time $\frac{11}{6}$-estimation algorithm~\cite{lars_2}.

Therefore, using the $2$-approximation algorithm, Theorem~\ref{theo:rob} yields
the following guarantee. For every fixed $\delta\in(0,1]$, one can compute a
schedule $\sigma^{\mathrm{rob}}$ in time $(nm)^{\bigO(1/\delta)}$ such that
$\frac{\Cmax(\sigma^{\mathrm{rob}})}{\OPT}
\le \min\left\{1+\delta\frac{\Err}{\OPT},\,2\right\}$.

Next, let $T$ be the makespan estimate returned by~\cite{lars_2}, rounded down to an integer. Then, $\OPT \le T \le \frac{11}{6}\OPT$.
If $\Cmax(\sigma^{\mathrm{rob}}) > T$, we can apply the Global Repair algorithm of Section~\ref{sec:repair-load} and get a feasible repaired schedule $\sigma^\mathrm{rep}$ with $\Cmax(\sigma^\mathrm{rep}) \le \frac{11}{6}\OPT
$
in running time
$\widetilde{\bigO}\!\left(
n \cdot m^{\bigO(\Err)}\cdot \Err^{\bigO(\Err)}
\right)$.
Thus, whenever the estimate $T$ improves over the robust guarantee above, the
repair procedure recovers a schedule matching the best known polynomial-time
estimation ratio in error-dependent time.

\subsection{Two-valued Restricted Assignment}

We now consider the two-valued special case in which each job has size either $1$ or $p_B > 1$. For this problem, Jansen et al. gave a polynomial-time estimation algorithm with ratio $\frac{5}{3} + \zeta_1$, for every fixed $\zeta_1 > 0$~\cite{Maack}. On the approximation side, Bamas et al.~\cite{bamas} designed a poly-time $(1.75 +\zeta_2)$-approximation algorithm, for every fixed $\zeta_2 > 0$.

Fix $\zeta_1,\zeta_2>0$. Then, by combining the smoothness algorithm
with the $(1.75+\zeta_2)$-approximation algorithm and returning the schedule
with smaller makespan, Theorem~\ref{theo:rob} gives the following guarantee.
For every fixed $\delta\in(0,1]$, one can compute a schedule
$\sigma^{\mathrm{rob}}$ in time $(nm)^{\bigO(1/\delta)}$ with
$\frac{\Cmax(\sigma^{\mathrm{rob}})}{\OPT}
\le
\min\left\{1+\delta\frac{\Err}{\OPT},\,1.75+\zeta_2\right\}$.

Now let $T$ be the makespan estimate returned by~\cite{Maack}, rounded down to an integer, so that $\OPT \le T \le \left(\frac{5}{3}+\zeta_1 \right)\OPT$.
Whenever the resulting schedule satisfies $\Cmax(\sigma^{\mathrm{rob}})>T$, we
may use the Global Repair algorithm from Section~\ref{sec:repair-load} to obtain
a feasible repaired schedule $\sigma^{\mathrm{rep}}$ with
$\Cmax(\sigma^{\mathrm{rep}})\le \left(\frac{5}{3}+\zeta_1 \right)\OPT$ in
running time
$\widetilde{\bigO}\!\left(
n \cdot m^{\bigO(\Err)}\cdot \Err^{\bigO(\Err)}
\right)$.

\begin{credits}
\subsubsection{\ackname}
The author is supported by a Humboldt Research Fellowship for Postdoctoral Researchers from the Alexander von Humboldt Foundation.

\subsubsection{\discintname}
The author declares no competing interests.
\end{credits}

\bibliographystyle{splncs04}
\bibliography{mybibliography}

\newpage
\appendix

\section{Repair Parameterized by the Number of Moved Jobs} \label{sec:repair_jobs}

In this section, we present a repair strategy whose running time depends on the number of moved jobs
rather than the moved-load. This can be significantly smaller when the predictor makes only a few mistakes
on large jobs: moving a single large job may contribute a large amount to moved-load, but it contributes only
one unit to the moved-job error. We first give such a repair strategy for \textsc{Restricted Assignment} and then
in Section~\ref{sec:2ra}, we refine it further for the two-valued version of the problem. Let us now formally define the moved-job error.

\begin{definition}[Moved-job error]
    For a feasible schedule $\sigma$ and any mapping $\tau:[n]\to[m]$, define:
$\moveJobs(\sigma,\tau)=|M(\sigma,\tau)| = | \{j\in[n]:\sigma(j)\neq \tau(j)\} |$.
\end{definition}

Next, we define the moved-job error of a predicted assignment $\hat\pi$ with respect to $\OPT$ and to a fixed makespan $T$.

\begin{definition}[Prediction moved-job error wrt $\OPT$]
Define the prediction error with respect to an optimal solution by
\[
 \ErrJ = \ErrJ(\OPT) = \min\{\moveJobs(\sigma^*,\hat\pi): \sigma^* \text{ feasible and } \Cmax(\sigma^*)=\OPT\}.
\]
\end{definition}

\begin{definition}[Prediction moved-job error wrt $T$]
Fix $T$. The minimum number of moved jobs needed to reach makespan at most $T$, relative to a predicted assignment $\hat\pi$, is
\[
\ErrJ(T) =\min\{\moveJobs(\sigma,\hat\pi): \sigma \text{ feasible and } \Cmax(\sigma)\le T\}.
\]
\end{definition}

Since $p_j \in \mathbb{Z}_{\ge 1}$ for all jobs, for every feasible schedule $\sigma$ we have that
\[
    \moveJobs(\sigma,\hat{\pi}) \le \move(\sigma,\hat{\pi}),
\]
because every moved job contributes at least one unit of processing time to the moved-load. Moreover, the analogue of Lemma~\ref{lem:lb-err(T)} for the moved-job error also holds, namely, if $T\ge \OPT$, then $\ErrJ(T)\le \ErrJ$.
Therefore, $\ErrJ(T)\le \Err(T)\le \Err$ and $\ErrJ \le \Err$.

The next lemma is the analogue of Lemma~\ref{lem:proj} for the moved-job error. Although the quantities $\ErrJ$ and $\ErrJ(T)$ are defined with respect to the original prediction $\hat{\pi}$, in the repair phase we work with the projected feasible schedule $\hat{\sigma}=\Proj(\hat{\pi})$. The lemma below shows that projecting the prediction cannot increase the number of moved jobs.

\begin{lemma}[Projection monotonicity (moved jobs)]\label{lem:proj-jobs}
For every feasible schedule $\sigma$,\\
$\moveJobs(\sigma,\hat{\sigma}) \le \moveJobs(\sigma,\hat{\pi})$, where $\hat{\sigma}=\Proj(\hat{\pi})$.
\end{lemma}

\begin{proof}
Fix a feasible schedule $\sigma$ and a job $j$.

If $\hat{\pi}(j)\in \Gamma(j)$, then $\hat{\sigma}(j)=\hat{\pi}(j)$, so
$ \mathbf{1}[\sigma(j)\neq \hat{\sigma}(j)] =
\mathbf{1}[\sigma(j)\neq \hat{\pi}(j)].$

If $\hat{\pi}(j)\notin \Gamma(j)$, then feasibility of $\sigma$ implies $\sigma(j)\in \Gamma(j)$, hence $\sigma(j)\neq \hat{\pi}(j)$. Therefore,
$
\mathbf{1}[\sigma(j)\neq \hat{\sigma}(j)]
\le
\mathbf{1}[\sigma(j)\neq \hat{\pi}(j)].
$
Summing over all jobs gives the claim. 
\qed \end{proof}

Next, we bound the number of machines touched by moved jobs as we did for moved-load.

\begin{lemma}[Incident set size bound under moved-job budget]
\label{lem:inc-size-jobs}
\leavevmode\\
If $\moveJobs(\sigma,\hat\sigma)\le M$, then $|A_{\mathrm{inc}}(\sigma,\hat\sigma)|\le 2M$.
\end{lemma}

\begin{proof}
If at most $M$ jobs are moved, then $|M(\sigma,\hat\sigma)|\le M$.
Each moved job involves at most two machines (its predicted machine and its assigned machine) to
$A_{\mathrm{inc}}$, hence
$|A_{\mathrm{inc}}(\sigma,\hat\sigma)|\le 2|M(\sigma,\hat\sigma)|\le 2M$. 
\qed \end{proof}

\subsection{Repair Oracle under Moved Jobs}
\label{subsec:dp-jobs-rounded-only}

We now give a fixed-set dynamic program, closely mirroring the one in Section~\ref{subsec:repair}. The crucial difference is that we apply a standard rounding step (as in FPTAS analyses for knapsack and related packing DPs), which reduces the state space to polynomial size and thereby removes any pseudo-polynomial dependence on the jobs' processing times. This is necessary to get an $\mathrm{XP}$ algorithm with respect to the number of moved jobs.

\begin{lemma}[Fixed-set Repair oracle (moved jobs)]
\label{lem:fixed-dp-jobs-rounded}
Fix $\varepsilon\in(0,1)$, a threshold makespan $T$, a moved-job budget $M\in\mathbb{Z}_{\ge 1}$,
and a feasible predicted schedule $\hat\sigma$.
Let $A\subseteq[m]$ satisfy $O(T)\subseteq A$. Write $A=\{i_1,\dots,i_a\}$ with $a=|A|\ge 1$, and let
$J_A = \{j\in[n]: \hat\sigma(j)\in A\}$.
There is a dynamic program (DP) running in time
$O\!\left(|J_A|\cdot a \cdot (2B+1)^{a-1}\right)$,
where  $B=\left\lceil \frac{3M^2}{\varepsilon}\right\rceil$,
that outputs either \textsc{Fail} or a feasible schedule $\sigma$ satisfying
\begin{enumerate}
\item $\sigma(j)=\hat\sigma(j)$ for all $j\notin J_A$;
\item $\sigma(j)\in \Gamma(j)\cap A$ for all $j\in J_A$;
\item $\moveJobs(\sigma,\hat\sigma)\le M$;
\item $\Cmax(\sigma)\le (1+\varepsilon)T$.
\end{enumerate}
Moreover, if there exists a feasible schedule $\sigma'$ satisfying (1)--(3) and $\Cmax(\sigma')\le T$,
then the DP always outputs a schedule (i.e., it does not output \textsc{Fail}).
\end{lemma}

\begin{proof}
We first round the processing times of the jobs.
Set
\[
\rho = \frac{\varepsilon T}{2M},
\qquad
w_j = \left\lceil \frac{p_j}{\rho}\right\rceil \in \mathbb{Z}_{\ge 1},
\qquad
p'_j = \rho w_j .
\]
Then for every job $j$,
\[
p_j \le p'_j < p_j+\rho
\qquad\text{and}\qquad
0 \le p'_j - p_j < \rho .
\]
Let $B = \left\lceil \frac{3M^2}{\varepsilon}\right\rceil$.
Since $T\ge \OPT$ implies $p_j\le T$ for all jobs, we have
$w_j \le \lceil T/\rho\rceil = \lceil 2M/\varepsilon\rceil$. Hence the total rounded inflow/outflow in $\rho$-units
contributed by at most $M$ moved jobs to any machine is at most
\[
M\cdot \left\lceil \frac{2M}{\varepsilon}\right\rceil
\le M\left(\frac{2M}{\varepsilon}+1\right)
\le \frac{3M^2}{\varepsilon}
\le B,
\]
where the second inequality uses $M \ge 1 > \varepsilon$.

Since $O(T)\subseteq A$, every machine $i\notin A$ satisfies $\hat L_i\le T$.
Under conditions~(1)--(2), the assignment on machines outside $A$ is unchanged, hence
$L_i(\sigma)=\hat L_i\le T$ for all $i\notin A$.
Therefore, among schedules satisfying~(1)--(2), it suffices to enforce the constraint $L_i(\sigma)\le (1+\varepsilon)T$ for machines in $A$.

For schedules satisfying~(1)--(2), define for each $r\in[a]$ the rounded net change
\[
\overline b_r = \sum_{j:\sigma(j)=i_r} w_j \;-\; \sum_{j:\hat\sigma(j)=i_r} w_j .
\]
Since jobs are only reassigned within $A$, we have $\sum_{r=1}^a \overline b_r=0$.

For any schedule satisfying condition~(3), fix $r$ and let $\overline{\Out}_r$ (resp.\ $\overline{\In}_r$) be the total rounded size
of moved jobs in $\rho$-units leaving (resp.\ entering) machine $i_r$. At most $M$ moved jobs can leave or enter $i_r$, and each moved job has
rounded size at most $\lceil 2M/\varepsilon\rceil$, hence
\[
\overline{\Out}_r\le M\Bigl\lceil \frac{2M}{\varepsilon}\Bigr\rceil \le B,
\qquad
\overline{\In}_r\le M\Bigl\lceil \frac{2M}{\varepsilon}\Bigr\rceil \le B.
\]
Therefore,
$|\overline b_r|=|\overline{\In}_r-\overline{\Out}_r|\le \max\{\overline{\In}_r,\overline{\Out}_r\}\le B$.
Thus it suffices to search over vectors with $\overline b_r\in[-B,B]$, for all $r \in [a]$.

We store $\overline x=(\overline b_1,\dots,\overline b_{a-1})$ and set the last coordinate implicitly as
$\overline b_a(\overline x)= -\sum_{r=1}^{a-1} \overline b_r$.

The state space is
$\overline{\mathcal{X}}
=\bigl\{\overline x\in[-B,B]^{a-1}\cap\mathbb{Z}^{a-1} : \overline b_a(\overline x)\in[-B,B]\bigr\}$.

Order $J_A=\{j_1,\dots,j_q\}$, where $q=|J_A|$.
For $\ell\in\{0,\dots,q\}$ and $\overline x\in\overline{\mathcal{X}}$, we let $\DP[\ell,\overline x]$ be the minimum number of moved jobs incurred by assigning the first $\ell$ jobs to machines in $\Gamma(\cdot)\cap A$ while inducing rounded net changes $\overline x$ on machines $i_1,\dots,i_{a-1}$:

\begin{multline*}
\DP[\ell,\overline x] =
\min \Bigl\{\#\{t\le \ell:\varphi(j_t)\neq \hat\sigma(j_t)\}:\\
\varphi(j_t)\in\Gamma(j_t)\cap A,\ \text{and induced rounded net change is }\overline x\Bigr\}.
\end{multline*}

Initialize $\DP[0,\vec 0]=0$ and $\DP[0,\overline x]=+\infty$ for $\overline x\neq \vec 0$.

Fix $\ell<q$, a state $\overline x\in\overline{\mathcal{X}}$ with $\DP[\ell,\overline x]<\infty$, and let $j=j_{\ell+1}$.
For each choice $u\in\Gamma(j)\cap A$, define the new state $\overline x'$ by updating $\overline x$ as follows.
Only the coordinates corresponding to $u$ and to $\hat\sigma(j)$ can change. Hence $\overline x'$ can be computed from $\overline x$
in $\bigO(1)$ time.
Formally, for each $r\in[a-1]$,
\[
\overline x'_r \;=\; \overline x_r \;+\; w_j\cdot\bigl(\Ind[u=i_r]-\Ind[\hat\sigma(j)=i_r]\bigr).
\]
The transition cost is
\[
\cost_J(j,u)=\begin{cases}
0 & \text{if } u=\hat\sigma(j),\\
1 & \text{otherwise.}
\end{cases}
\]
We keep the transition only if $\overline x'\in\overline{\mathcal{X}}$ and $\DP[\ell,\overline x]+\cost_J(j,u)\le M$, and then update
\[
\DP[\ell+1,\overline x'] \leftarrow \min\{\DP[\ell+1,\overline x'],\ \DP[\ell,\overline x]+\cost_J(j,u)\}.
\]

We accept a final state $(q,\overline x)$ if the implied rounded changes satisfy
\[
\hat L_{i_r} + \rho\cdot \overline b_r \le T + M\rho = \left(1+\frac{\varepsilon}{2}\right)T
\qquad \text{for all } r\in[a],
\]
where $(\overline b_1,\dots,\overline b_{a-1})=\overline x$ and $\overline b_a=\overline b_a(\overline x)$.
If such a final state exists, we reconstruct the corresponding schedule. Otherwise, we output \textsc{Fail}.

\paragraph*{Correctness.}
Let $\sigma$ be any schedule output by the DP, and fix a machine $i_r\in A$.
Write the true net change $b_r=L_{i_r}(\sigma)-\hat L_{i_r}$.
Jobs with $\sigma(j)=\hat\sigma(j)$ contribute equally to both schedules and cancel in the net change, hence only moved jobs contribute.
At most $M$ moved jobs can affect a given machine.
Each moved job contributes either $\pm p_j$ to $b_r$ and $\pm p'_j=\pm\rho w_j$ to $\rho\overline b_r$.
Since $|p'_j-p_j|<\rho$ for each job, we obtain
$|b_r-\rho \overline b_r|\le M\rho$.

Since the final state is accepted, $\hat L_{i_r}+\rho \overline b_r \le T+M\rho$ for all $r\in[a]$.
Therefore,
\[
L_{i_r}(\sigma)\le \hat L_{i_r}+\rho \overline b_r + M\rho \le T+2M\rho=(1+\varepsilon)T.
\]
Thus every accepted schedule has makespan at most $(1+\varepsilon)T$.

Conversely, suppose there exists a feasible schedule $\sigma'$ satisfying~(1)--(3) and $\Cmax(\sigma')\le T$.
Then for every $r\in[a]$ we have $\hat L_{i_r}+b_r'=L_{i_r}(\sigma')\le T$.
Using $|b_r'-\rho\overline b_r'|\le M\rho$ yields
\[
\hat L_{i_r}+\rho\overline b_r' \le \hat L_{i_r}+b_r'+M\rho \le T+M\rho,
\]
so $\sigma'$ satisfies the rounded acceptance constraints. Moreover, the DP can realize $\sigma'$ by choosing $u=\sigma'(j_t)$ at each step, hence it reaches an accepted state and does not output \textsc{Fail}.

\paragraph*{Running time.}
There are $q+1=|J_A|+1$ possible values of $\ell$, and for each value at most $|\overline{\mathcal{X}}|\le (2B+1)^{a-1}$ states.
From each state we try at most $|\Gamma(j_{\ell+1})\cap A|\le a$ transitions.
The next state $\overline x'$ can be computed in $\bigO(1)$ time as only two coordinates may change.
Hence the running time is
$O\!\left(|J_A|\cdot a \cdot (2B+1)^{a-1}\right)$.
\qed \end{proof}

Again, we wrap Lemma~\ref{lem:fixed-dp-jobs-rounded} into a global oracle by enumerating all candidate sets $A\subseteq[m]$
that could contain the incident machines of a solution moving at most $M$ jobs.
If $|O(T)|>2M$, then no such set exists and the oracle immediately returns \textsc{Fail}.

\begin{definition}[Global Repair oracle with $(T,M)$ (moved jobs)]
\label{def:oracle-jobs-rounded}
Given makespan $T$, budget $M\in\mathbb{Z}_{\ge 1}$, $\varepsilon\in(0,1)$ and a feasible predicted schedule $\hat\sigma$, the procedure
$\RepairOracleJ(T,M,\varepsilon,\hat\sigma)$ enumerates all sets $A\subseteq[m]$ such that $O(T)\subseteq A$ and
$|A|\le 2M$, runs the DP of Lemma~\ref{lem:fixed-dp-jobs-rounded} for each such $A$, and returns the first schedule found
(or \textsc{Fail} if none exists).
\end{definition}

\begin{lemma}[Correctness of the Global Repair oracle (moved jobs)]
\label{lem:oracle-jobs-rounded-correct}

If there exists a feasible schedule $\sigma'$ such that 
$\Cmax(\sigma')\le T $ and  
$\moveJobs(\sigma',\hat\sigma) \le M,$
then $\RepairOracleJ(T,M,\varepsilon,\hat\sigma)$ returns a feasible schedule
$\sigma$ with $\Cmax(\sigma)\le (1+\varepsilon)T$ and  $\moveJobs(\sigma,\hat\sigma)\le M$.
\end{lemma}

\begin{proof}
By Lemma~\ref{lem:over-incident}, we have $O(T)\subseteq A_{\mathrm{inc}}(\sigma',\hat\sigma)$.
By Lemma~\ref{lem:inc-size-jobs}, $|A_{\mathrm{inc}}(\sigma',\hat\sigma)|\le 2M$.
Hence the oracle enumerates $A=A_{\mathrm{inc}}(\sigma',\hat\sigma)$.
For this set $A$, the schedule $\sigma'$ satisfies the restrictions of Lemma~\ref{lem:fixed-dp-jobs-rounded}. Hence, by Lemma~\ref{lem:fixed-dp-jobs-rounded}, the fixed-set DP returns a schedule moving at most $M$ jobs and with makespan at most $(1+\varepsilon)T$.

\qed \end{proof}

\begin{lemma}[Running time of the Global Repair oracle (moved jobs)]
\label{lem:oracle-jobs-rounded-runtime}
Fix $T$, $M\in\mathbb{Z}_{\ge 1}$, $\varepsilon\in(0,1)$ and a feasible predicted schedule $\hat\sigma$.
The procedure $\RepairOracleJ(T,M,\varepsilon,\hat\sigma)$ runs in time
$O\!\left(n\cdot m^{2M}\cdot (2B+1)^{2M+1}\right)$, where $B=\left\lceil \frac{3M^2}{\varepsilon}\right\rceil$.
\end{lemma}

\begin{proof}
The oracle enumerates at most $\sum_{t=0}^{2M}\binom{m}{t}\le (2M+1)m^{2M}$ candidate sets.
For each such $A$, the DP (Lemma~\ref{lem:fixed-dp-jobs-rounded}) runs in time
\[
O\!\left(|J_A|\cdot |A|\cdot (2B+1)^{|A|-1}\right)
\le
O\!\left(n\cdot 2M\cdot (2B+1)^{2M-1}\right),
\]
since $|J_A|\le n$ and $|A|\le 2M$.
Multiplying the number of candidates by the worst-case DP time yields
\[
O\!\left(\Bigl(\sum_{t=0}^{2M}\binom{m}{t}\Bigr)\cdot n\cdot 2M\cdot (2B+1)^{2M-1}\right)
\le
O\!\left(n\cdot m^{2M}\cdot (2B+1)^{2M+1}\right),
\]
as claimed. 
\qed \end{proof}

\subsection{Global Repair under Moved Jobs}
\label{subsec:global-repair-jobs}

We now give a procedure to recover a schedule with makespan at most $(1+\varepsilon) T$ that does not require knowing the optimal moved-job budget $\ErrJ(T)$ in advance.

\begin{algorithm}[htbp]
\caption{\textsc{GlobalRepairJobs}$(T,\hat\pi,\varepsilon)$}
\label{alg:global-repair-jobs}
\begin{algorithmic}[1]
\Require Integer makespan threshold $T\ge \OPT$, possibly infeasible prediction $\hat\pi$, parameter $\varepsilon\in(0,1)$
\Ensure A feasible schedule $\sigma$ with $\Cmax(\sigma)\le (1+\varepsilon)T$
\State $\hat\sigma \gets \Proj(\hat\pi)$
\State Compute loads $\hat L_i$ and $\Delta(T)\gets \sum_{i\in[m]} (\hat L_i - T)^+$
\If{$\Delta(T)=0$}
  \State \Return $\hat\sigma$
\EndIf
\State $M \gets 1$
\While{\textbf{true}}
  \State $\sigma \gets \RepairOracleJ(T,M,\varepsilon,\hat\sigma)$
  \If{$\sigma \neq \textsc{Fail}$}
    \State \Return $\sigma$
  \EndIf
  \State $M \gets 2M$
\EndWhile
\end{algorithmic}
\end{algorithm}

\begin{theorem}[Global Repair at makespan $(1+\varepsilon)T$ (moved-job error)]
\label{thm:global-repair-jobs}
Fix any makespan threshold $T\ge \OPT$, let $\hat\sigma=\Proj(\hat\pi)$, and fix $\varepsilon\in(0,1)$.
Algorithm~\ref{alg:global-repair-jobs} returns a feasible schedule $\sigma$ with $\Cmax(\sigma)\le (1+\varepsilon)T$.
Moreover, if $\Delta(T)>0$, it stops for some integer $M$ satisfying:
$M < 2\ErrJ(T) \le 2\ErrJ$, and the number of oracle calls is at most $\lceil \log_2 \ErrJ(T)\rceil + 1 \le \lceil \log_2 \ErrJ\rceil + 1$.

If $\Delta(T)=0$, it directly outputs the predicted schedule $\hat\sigma$.

\end{theorem}

\begin{proof}
If $\Delta(T)=0$, the algorithm directly outputs $\hat{\sigma}$ and $\Cmax(\hat\sigma) \le T \le (1+\varepsilon)T$.

Assume $\Delta(T)>0$. Since $T\ge \OPT$, a feasible schedule with makespan at most $T$ exists, hence $\ErrJ(T)<\infty$. Moreover, $\Delta(T) >0$ implies $\ErrJ \ge \ErrJ(T) \ge1$.
The algorithm calls the oracle for $M=1,2,4,\dots$ until success.
The doubling sequence reaches a first value $M^*\ge \ErrJ(T)$ after at most $\lceil\log_2 \ErrJ(T)\rceil+1$ iterations, and by minimality
$M^*<2\ErrJ(T)\le 2\ErrJ$.
We show that the oracle succeeds when called with budget $M^*$. Hence the algorithm stops no later than this call.

For the budget $M^*$, by definition of $\ErrJ(T)$ there exists a feasible schedule $\sigma'$ with
\[
\Cmax(\sigma')\le T
\quad\text{and}\quad
\moveJobs(\sigma',\hat\pi)=\ErrJ(T)\le M^*.
\]
By Lemma~\ref{lem:proj-jobs}, we have
$\moveJobs(\sigma',\hat\sigma)\le \moveJobs(\sigma',\hat\pi)\le M^*$.
Hence, by Lemma~\ref{lem:oracle-jobs-rounded-correct}, the oracle succeeds when called with budget $M^*$ and returns a feasible schedule $\sigma$ with
$\Cmax(\sigma)\le (1+\varepsilon)T$.
Therefore, Algorithm~\ref{alg:global-repair-jobs} returns such a schedule.

Since the algorithm stops no later than the call with budget $M^*$, the budget $M$ at which it stops satisfies $M\le M^*<2\ErrJ(T)\le 2\ErrJ$.
The bound on the number of oracle calls follows from the length of the doubling sequence.
\qed \end{proof}

Consequently, since the successful oracle call uses $M < 2\ErrJ$ we get the following corollary.

\begin{corollary}[Running time of \textsc{GlobalRepairJobs}] \label{cor:time-jobs}
Assume that $T\ge \OPT$ and $\Delta(T)>0$. Then Algorithm~\ref{alg:global-repair-jobs} returns a feasible schedule $\sigma$ with
$\Cmax(\sigma)\le (1+\varepsilon)T$
in time
$\widetilde{\bigO}\!\left(
n \cdot m^{\bigO(\ErrJ)} \cdot
\left(\frac{\ErrJ}{\varepsilon}\right)^{\bigO(\ErrJ)}
\right)$.
\end{corollary}

\subsection{Exact Moved-Job Repair for Two-Valued Restricted Assignment}
\label{sec:2ra}

We now show that for the two-valued case ($p_j \in \{1,p_B\}$, with integer $p_B>1$), the $(1+\varepsilon)$-relaxation in Theorem~\ref{thm:global-repair-jobs} can be removed.

\begin{lemma}[Fixed-set exact repair oracle] \label{lem:fixed-set-exact}
Fix the two-valued case $p_j \in \{1,p_B\}$. Fix an integer threshold makespan $T$, a moved-job
budget $M \in \mathbb{Z}_{\ge 1}$, and a feasible predicted schedule $\hat{\sigma}$.
Let $A \subseteq [m]$ satisfy $O(T) \subseteq A$. Write $A=\{i_1,\dots,i_a\}$ with
$a=|A|\ge 1$, and define
$J_A = \{j\in [n] : \hat{\sigma}(j)\in A\}$.
There exists a procedure running in time
$O\!\left(|J_A|\cdot a \cdot (2M+1)^{a-1}
\;+\; (2M+1)^{a-1}\cdot \Phi_{\text{flow}}(n,m)\right)$,
where $\Phi_{\text{flow}}(n,m)$ is the worst-case running time of one min-cost integral flow instance,
that decides whether there exists a feasible schedule $\sigma$ such that:
\begin{enumerate}
    \item $\sigma(j)=\hat{\sigma}(j)$ for all $j\notin J_A$;
    \item $\sigma(j)\in \Gamma(j)\cap A$ for all $j\in J_A$;
    \item $\mathrm{moveJobs}(\sigma,\hat{\sigma})\le M$;
    \item $\Cmax(\sigma)\le T$.
\end{enumerate}
If such a schedule exists, the procedure outputs one. Otherwise it outputs \textsc{Fail}.
\end{lemma}

\begin{proof}
Partition $J_A = J_A^b \cup J_A^s$ into big and small jobs.

Since $O(T)\subseteq A$, every machine $i\notin A$ satisfies $\hat{L}_i\le T$.
Under conditions~(1)--(2), the assignment outside $A$ is unchanged, hence
$L_i(\sigma)=\hat{L}_i\le T$ for all $i\notin A$.
Therefore, among schedules satisfying (1)--(2), it suffices to enforce $L_i(\sigma)\le T$ for all $i\in A$.

\paragraph*{Phase 1 (Big-job DP).}
We construct a DP over $J_A^b$ only.
Since every big job has size $p_B$, the net load change $b_r$ on each machine $i_r\in A$
due to big-job reassignment is a multiple of $p_B$.
Since $\moveJobs(\sigma,\hat{\sigma})\le M$, the number of moved big jobs is at most $M$,
so $b_r/p_B \in [-M,M]\cap \mathbb{Z}$.
We therefore store
\[
x_r = b_r/p_B,
\]
giving a state space of size $(2M+1)^{a-1}$.

As in Lemma~\ref{lem:fixed-dp-load}, we store only $x=(x_1,\dots,x_{a-1})$ and define the last coordinate
implicitly by
$x_a(x)=-\sum_{r=1}^{a-1} x_r$.

Order $J_A^b = \{j_1,\dots,j_q\}$.
For $\ell\in\{0,\dots,q\}$ and each state $x$, the DP stores the minimum number of moved
big jobs incurred by assigning the first $\ell$ jobs to machines in $\Gamma(\cdot)\cap A$
while inducing net changes $x$ on machines $i_1,\dots,i_{a-1}$.
Transitions are exactly as in Lemma~\ref{lem:fixed-dp-jobs-rounded}, but in $p_B$-units:
for job $j$ and destination $u\in \Gamma(j)\cap A$, the state is updated by
\[
x_r' = x_r + \mathbf{1}[u=i_r]-\mathbf{1}[\hat{\sigma}(j)=i_r]
\qquad \forall r\in [a-1],
\]
and the transition cost is
\[
\cost_J(j,u)=
\begin{cases}
0 & \text{if } u=\hat{\sigma}(j),\\
1 & \text{otherwise.}
\end{cases}
\]

A final state $(q,x)$ is accepted if the implied big-job loads satisfy
\[
\hat{L}_{i_r}^b + p_B x_r \le T
\qquad \text{for all } r\in [a],
\]
where $x_a=x_a(x)$, $\hat{L}_{i_r}^b$ is the load of big jobs on machine $i_r$ according to $\hat\sigma$, and if the DP value at $(q,x)$ is at most $M$.
If an accepted state exists, let $\sigma^b$ be a reconstructed big-job assignment for that state
with moved-job cost equal to the stored DP value $M_b$.
Define the residual capacity/load on each machine $i_r\in A$ by
\[
R_{i_r} = T-\sum_{j\in J_A^b:\,\sigma^b(j)=i_r} p_B .
\]

\paragraph*{Phase 2 (Small-job min-cost flow).}
For each accepted big-job state, we construct a flow network as follows.
The source sends one unit to each job $j\in J_A^s$.
Each such job $j$ connects to every machine $i\in \Gamma(j)\cap A$ with capacity $1$ and cost equal to
$\mathbf{1}[i\neq \hat{\sigma}(j)]$.
Each machine $i_r$ connects to the sink with capacity $\min\{R_{i_r},|J_A^s|\}$.
We compute a minimum-cost integral flow of value $|J_A^s|$. By integrality, this flow induces an assignment of each job in $J_A^s$ to a single machine in $\Gamma(j)\cap A$.
If for some accepted big-job state this flow has cost at most $M-M_b$, then combining it with
the corresponding $\sigma^b$ yields a schedule satisfying (1)--(4), which we output.

If no accepted big-job state yields such a flow, we output \textsc{Fail}.

\paragraph*{Correctness.}
Any output schedule clearly satisfies (1) and (2).
Its moved-job count is $M_b$ plus the flow cost, hence at most $M$, so (3) holds.
For machines outside $A$, the load remains unchanged and is at most $T$.
For a machine $i_r\in A$, the big-job load is exactly $T-R_{i_r}$, and the flow sends at most
$R_{i_r}$ unit-size small jobs to $i_r$.
Thus its final load is at most $T$, proving (4).

Conversely, suppose there exists a feasible schedule $\sigma'$ satisfying (1)--(4).
Let $M_b'$ be the number of moved big jobs in $\sigma'$.
Then $\sigma'|_{J_A^b}$ defines a valid path in the big-job DP constructed above, and the resulting final
state is accepted because the big-job load on each machine is at most its total load under
$\sigma'$, hence at most $T$.

Let $x$ be the accepted terminal state induced by $\sigma'|_{J_A^b}$, and let $M_b$ be the DP
value stored at $x$.
By definition of the DP, we have $M_b \le M_b'$.
Moreover, for a fixed accepted state $x$, the quantities
$\hat{L}_{i_r}^b + p_B x_r$ and hence $R_{i_r} = T-(\hat{L}_{i_r}^b + p_B x_r)$
depend only on $x$, not on the particular reconstructed big-job assignment for that state. Since $x$ is the state induced by $\sigma'|_{J_A^b}$, we have
\[
R_{i_r} = T-L_{i_r}^b(\sigma')
\qquad \text{for all } r\in [a].
\]
Since small jobs have size $1$ and $L_{i_r}(\sigma')\le T$, the number of small jobs assigned by
$\sigma'$ to $i_r$ is at most $R_{i_r}$.
Hence $\sigma'|_{J_A^s}$ defines a feasible flow of value $|J_A^s|$, and its cost is exactly the
number of moved small jobs:
\[
\moveJobs(\sigma',\hat{\sigma})-M_b'
\le M-M_b'
\le M-M_b.
\]
Therefore the procedure succeeds whenever a schedule satisfying (1)--(4) exists.

\paragraph*{Running time.}
The DP has $|J_A^b|+1$ possible values of $\ell$, and for each value at most $(2M+1)^{a-1}$ states, with at most $a$
transitions per state, giving time 
$O\!\left(|J_A|\cdot a \cdot (2M+1)^{a-1}\right)$.

For each accepted terminal state, we use an exact algorithm to solve the min-cost flow instance that arises in time $\Phi_{\text{flow}}(n,m)$, which is polynomial in
$n,m$. Therefore, the total additional time is
$(2M+1)^{a-1}\cdot \Phi_{\text{flow}}(n,m)$.
This proves the claimed running time. 
\qed \end{proof}

We can now define $\ExactRepairOracle_J(T,M,\hat\sigma)$ with the fixed-set oracle of Lemma~\ref{lem:fixed-set-exact} as we did for $\RepairOracleJ$ in Definition~\ref{def:oracle-jobs-rounded}. Then we can use Algorithm~\ref{alg:global-repair-jobs} with $\ExactRepairOracle_J$ and define \textsc{ExactGlobalRepairJobs}.

By following closely the proofs of Lemma~\ref{lem:oracle-jobs-rounded-correct},
Lemma~\ref{lem:oracle-jobs-rounded-runtime}, and
Theorem~\ref{thm:global-repair-jobs}, and replacing
$\RepairOracleJ$ by $\ExactRepairOracle_J$, we obtain similar conclusions with the stronger
guarantee $\Cmax(\sigma)\le T$.
For the min-cost flow computation in Phase~2, the underlying graph has at most $\bigO(n+m)$ vertices
and $\bigO(nm)$ edges, with integral demands, costs in $\{0,1\}$, and capacities at most $n$. Using the exact min-cost flow algorithm of~\cite{flow2}, such an instance can be solved in $(nm)^{1+o(1)} \log_2 n$ time.

Consequently, in the two-valued case we obtain the following corollary.
\begin{corollary}[Running time of \textsc{ExactGlobalRepairJobs}] \label{cor:time-exact-jobs}
Assume the two-valued case $p_j\in\{1,p_B\}$, $T\ge \OPT$ and $\Delta(T)>0$. Then \textsc{ExactGlobalRepairJobs} (Algorithm~\ref{alg:global-repair-jobs} with
$\ExactRepairOracle_J$) returns a feasible schedule
$\sigma$ with $\Cmax(\sigma)\le T$ in time
$\widetilde{\bigO}\!\left(
(nm)^{1+o(1)} \cdot m^{\bigO(\ErrJ)} \cdot \ErrJ^{\bigO(\ErrJ)}
\right)$.
\end{corollary}

\section{Deferred Proofs for Moved-Load Repair}
\label{app:moved-load-details}

\subsection{Proof of Lemma~\ref{lem:fixed-dp-load}}

\begin{proof}
Since $O(T)\subseteq A$, every machine $i\notin A$ is not overloaded in the predicted schedule, i.e., $\hat L_i\le T$.
Under conditions~(1)--(2), the assignment on machines outside $A$ is unchanged, hence
$L_i(\sigma)=\hat L_i\le T$ for all $i\notin A$.
Therefore, among schedules satisfying~(1)--(2), the constraint $\Cmax(\sigma)\le T$ is equivalent to requiring $ L_i(\sigma)\le T$, for all $i\in A$.

Write $A=\{i_1,\dots,i_a\}$ and recall $J_A=\{j:\hat\sigma(j)\in A\}$.
For any schedule $\sigma$ satisfying~(1)--(2), define the net load change on each machine $i_r\in A$ by
\[
b_r = L_{i_r}(\sigma)-\hat L_{i_r}\qquad \text{for all } r\in[a].
\]
Because only jobs in $J_A$ are reassigned and they remain within $A$, the total load inside $A$ is preserved, and thus $\sum_{r=1}^a b_r = 0$.

Moreover, if $\move(\sigma,\hat\sigma)\le K$, then for each $r$ let $\Out_r$ be the total processing time of jobs that leave $i_r$
and $\In_r$ the total processing time of jobs that enter $i_r$.
Both $\Out_r$ and $\In_r$ are bounded by the total moved-load, so $\Out_r\le K$ and $\In_r\le K$, and thus
\[
|b_r|=|\In_r-\Out_r|\le K \qquad \text{for all } r\in[a].
\]
Consequently, it suffices to search over integer vectors $(b_1,\dots,b_a)$ with $\sum_r b_r=0$ and $b_r\in[-K,K]$, since $p_j \in \mathbb{Z}_{\ge 1}$ for all jobs.

We store $x=(b_1,\dots,b_{a-1})$, with $b_a(x)=-\sum_{r=1}^{a-1}b_r$.

The state space is
$\mathcal{X}=\bigl\{x\in[-K,K]^{a-1}\cap\mathbb{Z}^{a-1} : b_a(x)\in[-K,K]\bigr\}$.

Order $J_A=\{j_1,\dots,j_q\}$, where $q=|J_A|$.
For $\ell\in\{0,\dots,q\}$ and $x\in\mathcal{X}$, we let $\DP[\ell,x]$ be the minimum moved-load incurred by assigning the first $\ell$
jobs to machines in $\Gamma(\cdot)\cap A$ while inducing net changes $x$ on machines $i_1,\dots,i_{a-1}$:
\begin{multline*}
\DP[\ell,x] \;=\;
\min \Bigl\{\sum_{t\le \ell:\,\varphi(j_t)\neq \hat\sigma(j_t)} p_{j_t} :
\varphi(j_t)\in\Gamma(j_t)\cap A,\ \text{and induced net change is } x\Bigr\}.
\end{multline*}
Initialize $\DP[0,\vec 0]=0$ and $\DP[0,x]=+\infty$ for $x\neq \vec 0$.

The net load change on machine $i_r$ is $+p_j$ if the job is assigned to $i_r$ but was not predicted there, $-p_j$ if the job was predicted on $i_r$ but is reassigned elsewhere, and $0$ if the job remains on its predicted machine $i_r$. Effects involving $i_a$ are handled implicitly through $b_a(x)=-\sum_{r\le a-1}x_r$. More formally, we have the following transitions.

Fix $\ell<q$, a state $x\in\mathcal{X}$ with $\DP[\ell,x]<\infty$, and let $j=j_{\ell+1}$.
For each choice $u\in\Gamma(j)\cap A$, define the new state $x'$ by updating each coordinate $r\in[a-1]$ as
\[
x'_r \;=\; x_r \;+\; p_j\cdot\bigl(\mathbf{1}[u=i_r]-\mathbf{1}[\hat\sigma(j)=i_r]\bigr).
\]

The transition cost is
\[
\cost(j,u)=\begin{cases}
0 & \text{if } u=\hat\sigma(j),\\
p_j & \text{otherwise.}
\end{cases}
\]
We keep the transition only if $x'\in\mathcal{X}$ and $\DP[\ell,x]+\cost(j,u)\le K$, and then update
\[
\DP[\ell+1,x'] \leftarrow \min\{\DP[\ell+1,x'],\ \DP[\ell,x]+\cost(j,u)\}.
\]

We accept a final state $(q,x)$ if the net-change vector $(b_1,\dots,b_a)$ satisfies the load constraints $\hat L_{i_r}+b_r \le T$, for all  $r\in[a]$,
where $(b_1,\dots,b_{a-1})=x$ and $b_a=b_a(x)$.
If such a final state exists, we reconstruct the corresponding schedule $\sigma$. Otherwise, we output \textsc{Fail}.

\paragraph*{Correctness.}
Any schedule $\sigma$ satisfying~(1)--(2) corresponds to a unique DP path: at step $t$ choose $u=\sigma(j_t)$.
The state update matches exactly the induced net changes on $i_1,\dots,i_{a-1}$, and the accumulated cost equals the moved-load of
the processed jobs, hence the final DP value equals $\move(\sigma,\hat\sigma)$.
Conversely, any DP path defines an assignment of all jobs in $J_A$ into $\Gamma(\cdot)\cap A$. Together with~(1) it yields a full schedule.
Therefore, the DP reaches an accepted state with value $\le K$ if and only if there exists a schedule satisfying~(1)--(4).

\paragraph*{Running time.}
There are $q+1=|J_A|+1$ possible values of $\ell$, and for each value at most $|\mathcal{X}|\le (2K+1)^{a-1}$ states.
From each state, we try at most $|\Gamma(j_{\ell+1})\cap A|\le a$ transitions, each in $\bigO(1)$ time. Hence the running time is
$O\!\left(|J_A|\cdot a \cdot (2K+1)^{a-1}\right)$. 
\qed \end{proof}

\subsection{Proof of Lemma~\ref{lem:oracle-correct}}

\begin{proof}
By Lemma~\ref{lem:over-incident}, we have
$O(T)\subseteq A_{\mathrm{inc}}(\sigma',\hat\sigma)$.
By Lemma~\ref{lem:inc-size}, $|A_{\mathrm{inc}}(\sigma',\hat\sigma)|\le 2K$.
Hence the oracle enumerates $A=A_{\mathrm{inc}}(\sigma',\hat\sigma)$, and for this set $A$, the schedule $\sigma'$ satisfies the restrictions of Lemma~\ref{lem:fixed-dp-load}.
Therefore, the DP succeeds and returns a schedule with the required guarantees. 
\qed \end{proof}

\subsection{Proof of Lemma~\ref{lem:oracle-runtime-load}}

\begin{proof}
The number of candidate sets $A$ with $O(T)\subseteq A$ and $|A|\le 2K$ is at most $\sum_{t=0}^{2K}\binom{m}{t}\le (2K+1)m^{2K}$. For each, the DP of Lemma~\ref{lem:fixed-dp-load} runs in time $\bigO(n\cdot 2K\cdot (2K+1)^{2K-1})$, since $|J_A|\le n$ and $|A|\le 2K$. Multiplying gives the bound. 
\qed \end{proof}

\subsection{Proof of Theorem~\ref{thm:parametric-repair}}

\begin{proof}
If $\Delta(T)=0$, then $\hat L_i\le T$ for all $i$. Hence $\Cmax(\hat\sigma)\le T$ and the algorithm returns $\hat\sigma$.

Assume $\Delta(T)>0$. By Proposition~\ref{prop:lb-chain}, we have $\Err(T)\ge \Delta(T)$.
Thus, the doubling sequence
$K=\Delta(T),\ 2\Delta(T),\ 4\Delta(T),\ \dots$
reaches a first value $K^*$ with $K^*\ge \Err(T)$ after at most
$\left\lceil \log_2(\Err(T)/\Delta(T))\right\rceil+1$ iterations, and by minimality $K^*<2\Err(T)$.
We show that the oracle succeeds when called with budget $K^*$. Hence the algorithm stops no later than this call.

For the budget $K^*$, by definition of $\Err(T)$ there exists a feasible schedule $\sigma'$ with
\[
\Cmax(\sigma') \le T
\quad\text{and}\quad
\move(\sigma',\hat{\pi}) = \Err(T) \le K^*.
\]
By Lemma~\ref{lem:proj}, we have
$\move(\sigma',\hat{\sigma}) \le \move(\sigma',\hat{\pi}) \le K^*$.
Therefore the hypothesis of Lemma~\ref{lem:oracle-correct} is satisfied, so the oracle succeeds when called with budget $K^*$ and returns a feasible schedule $\sigma$ with $\Cmax(\sigma) \le T$.

Since the algorithm stops no later than the call with budget $K^*$, the budget $K$ at which it stops satisfies $K\le K^*<2\Err(T)$. Finally, since $T\ge \OPT$, Lemma~\ref{lem:lb-err(T)} yields $\Err(T)\le \Err$, and hence $K<2\Err(T)\le 2\Err$.
The second bound on the number of oracle calls follows similarly.
\qed \end{proof}

\section{Proof of the Hardness Result}
\label{app:repair-hardness}

\begin{proof}[Proof of Theorem~\ref{thm:exact-moved-load-repair-hardness}]
We reduce from \textsc{Multicolored Clique}. We are given a $k$-partite graph $G=(V_1\cup\cdots\cup V_k,E)$, with $k \ge 2$, and ask whether there are vertices $v_i\in V_i$, for all $i\in[k]$, such that $\{v_i,v_j\}\in E$ for every $i\ne j$. This problem is $\mathrm{W[1]}$-hard parameterized by $k$~\cite{cygan}.

Given an instance $G=(V_1\cup\cdots\cup V_k,E)$, we construct an instance of
\textsc{Exact Moved-Load Repair}. We create a vertex machine $M_v$ for every
vertex $v$, an edge machine $M_e$ for every edge $e$, a pair-source machine
$P_{ij}$ for every pair $1\le i<j\le k$, and one sink machine $S$.

Let $T = 2|E|+k(k-1)$.
We construct a feasible predicted schedule $\hat\sigma$ with loads
\[
    L_{M_v}(\hat\sigma)=T,\qquad
    L_{M_e}(\hat\sigma)=T,\qquad
    L_{P_{ij}}(\hat\sigma)=T+2,\qquad
    L_S(\hat\sigma)=T-k(k-1).
\]

The non-dummy jobs are as follows. For every vertex $v$, create a blocker job
$b_v$ of size $k-1$, assigned in $\hat\sigma$ to $M_v$, with
$\Gamma(b_v)=\{M_v,S\}$.

For every edge $e=\{u,v\}$, create two endpoint jobs $a_{e,u}$ and $a_{e,v}$ of
size $1$, both assigned in $\hat\sigma$ to $M_e$, with
\[
    \Gamma(a_{e,u})=\{M_e,M_u\},
    \qquad
    \Gamma(a_{e,v})=\{M_e,M_v\}.
\]
For every color pair $i<j$ and every edge $e$ between $V_i$ and $V_j$, create a
pair-edge job $q_{ij,e}$ of size $2$, assigned in $\hat\sigma$ to $P_{ij}$, with
$\Gamma(q_{ij,e})=\{P_{ij},M_e\}$.

Finally, add unit-size dummy jobs, each eligible only on its predicted machine,
so that the loads above are achieved. By the definition of $T$, the required
dummy loads are nonnegative. Set $K=3k(k-1)$.

The total overload of the pair-source machines above $T$ is
$2\binom{k}{2}=k(k-1)$, exactly the slack of $S$. All vertex and edge machines have load $T$. Hence the total load is $|\mathcal M|T$, where $\mathcal M$ is the set of
machines, and therefore any schedule of makespan at most $T$ has every machine
at load exactly $T$.

Suppose first that $G$ has a multicolored clique with vertices
$v_1,\ldots,v_k$. For every $i$, move $b_{v_i}$ to $S$. For every pair $i<j$,
let $e_{ij}=\{v_i,v_j\}$, move $q_{ij,e_{ij}}$ to $M_{e_{ij}}$, and move the two
endpoint jobs $a_{e_{ij},v_i}$ and $a_{e_{ij},v_j}$ to $M_{v_i}$ and
$M_{v_j}$, respectively. The moved load is
\[
    k(k-1)+2\binom{k}{2}+2\binom{k}{2}=3k(k-1)=K.
\]
The load of each $P_{ij}$ decreases by $2$, each used edge machine receives
load $2$ and sends out load $2$, each selected vertex machine sends out its
blocker of size $k-1$ and receives $k-1$ endpoint jobs, and $S$ receives total
load $k(k-1)$. Thus every machine ends with load $T$.

Conversely, suppose there is a feasible schedule $\sigma$ with
$\Cmax(\sigma)\le T$ and $\move(\sigma,\hat\sigma)\le K$. As argued above,
every machine has final load exactly $T$. Each pair-source machine $P_{ij}$
starts with load $T+2$, no job can move into it, and its movable jobs all have
size $2$. Thus, exactly one pair-edge job leaves it. Let this job be
$q_{ij,e_{ij}}$. Then $M_{e_{ij}}$ receives load $2$ and must send out exactly
two units of load. The only movable jobs initially on $M_{e_{ij}}$ are its two
unit-size endpoint jobs, so both endpoint jobs of $e_{ij}$ move to their
corresponding vertex machines.

Fix a color class $V_i$. Exactly one selected edge is chosen for each of the
$k-1$ color pairs involving $i$, and hence exactly $k-1$ endpoint jobs are sent
to vertex machines of color $i$. Let $t_v$ be the number of such endpoint jobs
received by $M_v$, for $v\in V_i$. Then $\sum_{v\in V_i}t_v=k-1$. Since every
vertex machine starts and ends with load $T$, if $t_v>0$, then $M_v$ must lose
exactly $t_v$ units of load. The only movable job initially on $M_v$ is the
blocker $b_v$ of size $k-1$, so $t_v\in\{0,k-1\}$. Therefore there is a unique
vertex $v_i\in V_i$ with $t_{v_i}=k-1$.

For every pair $i<j$, the endpoint jobs of $e_{ij}$ move to the vertex machines
corresponding to the endpoints of $e_{ij}$. Since all endpoint jobs sent to
color class $V_i$ go to $M_{v_i}$, and all endpoint jobs sent to color class
$V_j$ go to $M_{v_j}$, we have $e_{ij}=\{v_i,v_j\}$. Thus
$v_1,\ldots,v_k$ form a multicolored clique. 
Since $K=3k(k-1)$ depends only on $k$, the theorem follows.
\qed \end{proof}

The same construction also gives hardness for the moved-job budget.

\begin{corollary}
\label{cor:exact-moved-job-repair-hardness-clique}
\textsc{Exact Moved-Job Repair} for \textsc{Restricted Assignment} is
$\mathrm{W[1]}$-hard parameterized by the moved-job budget $M$, even when
$\hat\sigma$ is feasible.
\end{corollary}

\begin{proof}
Use the same reduction and set
$M=k+3\binom{k}{2}$.

In the forward direction, the repair moves exactly $k$ blocker jobs,
$\binom{k}{2}$ pair-edge jobs, and $2\binom{k}{2}$ endpoint jobs, for a total
of $M$ moved jobs.

Conversely, any feasible repair of makespan at most $T$ must make exactly the
same forced moves as in the proof of
Theorem~\ref{thm:exact-moved-load-repair-hardness}: one pair-edge job leaves
each $P_{ij}$, the two corresponding endpoint jobs leave the selected edge
machine, and exactly $k$ blocker jobs move to the sink. These already account
for $k+3\binom{k}{2}=M$ moved jobs, so no additional job can be moved. The same
consistency argument then yields a multicolored clique.
\qed \end{proof}

\end{document}